\documentclass[10pt,conference]{IEEEtran}
\usepackage[utf8]{inputenc}
\usepackage{amsmath,amssymb,amsthm}
\usepackage{booktabs}
\usepackage{graphicx}
\usepackage{url}
\usepackage[hidelinks]{hyperref}
\usepackage{microtype}
\usepackage{balance}
\usepackage{multirow}

\newtheorem{proposition}{Proposition}
\newtheorem{definition}{Definition}

\begin{document}

\title{When Is a Reconciliation Break Real?\\
A Temporal Correctness Contract for Streaming Financial Reconciliation}

\author{\IEEEauthorblockN{Abhishek Sharma}
\IEEEauthorblockA{Senior Member, IEEE\\
\texttt{abhicse24@gmail.com} \quad ORCID: 0009-0007-1103-2103}}

\maketitle

\begin{abstract}
Reconciliation systems decide, continuously, whether money the ledger expected
has arrived at the bank. On an asynchronous stream that decision is taken
against partial evidence: a credit that has not appeared may be missing, or
may simply not have been posted yet. Deployed pipelines settle this with a
timeout or a watermark, and the exceptions are treated as findings. We give a
temporal correctness contract that says when such a verdict is admissible,
relative to a declared per-source lateness bound, and a checker for it that
needs no ground truth. The contract separates three things an operator
otherwise sees only as a rising break rate: the reconciler decided too early,
an input broke the bound it declared, or the money really is gone. On
workloads from a published payment simulator, the standard remedy for a
stalled streaming join, an idleness timeout, introduces contract violations on
5.42\% of obligations (95\% CI [3.94, 6.92]) on traffic where the same job
without it is exactly sound, and violations vanish once the timeout reaches
the slow source's posting period, as predicted. Under a lateness shift, all
four identity-based reconcilers we test raise spurious breaks on between
18.3\% and 23.0\% of the credits the bank did settle, while their contract
violation rate stays at or below 1.01\%: the checker attributes the loss to
240 bank-side and 413 ledger-side bound violations rather than to the
reconcilers. Tracking one horizon per source leaves break latency 12.7\% lower
on average and halves its tail (p99 224 against 446 ticks) at equal soundness,
though it is not cheaper at every cost setting and we report where it loses.
An Apache Flink interval join, written the ordinary way, violates the contract
on 14.9--20.1\% of obligations. The contract is expressible as SQL over the
tables a reconciliation system already keeps, agreeing with the reference
implementation on all 4{,}395 verdicts.
\end{abstract}


\section{Introduction}

A payment settles, and two systems find out at different times. The ledger
learns within seconds, because the processor told it. The bank feed posts on a
cycle, so the same fact can take an hour to appear. In between, a reconciliation system looking for the credit sees nothing, and has to decide
what nothing means.

Most of the time it means the batch has not run. Occasionally it means the money
is gone. The two are indistinguishable from the evidence available at the
moment the question is asked, and the distinction is the whole job.

What deployed systems do is pick a number. Wait ninety minutes, then raise an
exception; or set a watermark and let the stream processor decide when the
window has closed. Both are defensible engineering and both are, as
correctness criteria, unstated: the number encodes an assumption about how late
the bank can be, that assumption is nowhere written down, and when the break
rate doubles there is no way to tell whether the reconciler became wrong or the
bank became slow. The exception queue grows either way, and the team retunes the
timeout either way.

This paper writes the assumption down and makes it checkable. The contribution
is a temporal correctness contract for reconciliation on asynchronous streams:
a statement of when a reconciler is \emph{entitled} to declare a match or a
break, given only what it has seen and a declared bound on how late each source
may be. Entitlement is not accuracy. A reconciler can be entitled and wrong,
which happens when an input violates the bound it declared, and it can be
unentitled and right, which happens whenever the traffic was benign. Those two
cases have different owners and different fixes, and telling them apart does not require knowing the answer, which is what
makes the check deployable.

Three findings follow from the contract, and one of them is a negative result
about our own construction.

The documented fix for a stalled streaming join is unsound for reconciliation.
A merged watermark takes the minimum across inputs, so a source that posts in
batches freezes the operator between batches; the remedy is an idleness timeout
that drops the quiet input from the minimum. For a reconciliation join the
quiet input is the one the decision is about. We show the resulting violations
appear exactly when the idleness timeout falls below the posting period and not
above it, on traffic where the same job is otherwise exactly sound.

Under a shift in the lateness distribution, entitlement and accuracy come
apart cleanly. Every arm we test raises spurious breaks on between 18.3\% and
23.0\% of the credits that did settle, while its contract violation rate stays
at or below 1.01\%. The checker
attributes the degradation to the inputs, and it does so from the stream and
the declared model alone.

Tracking a horizon per source is better on latency and not uniformly better on
cost. At matched soundness it halves the tail of break latency and reduces its
standard deviation from 60.8 to 5.0 ticks, because a timeout measured from the
ledger's report inherits the ledger's own jitter into a decision that is about
the bank. At the calibrated operating point, on the other hand, it decides
earlier and is therefore more exposed when bounds are violated: over a
$5\times5\times5$ sweep of the cost constants it is cheapest on all 125 points
in calm traffic and on only 30 of 125 under faults. We report both.

\section{Background and related work}

\subsection{Monitoring temporal properties over out-of-order streams}

Basin, Klaedtke and Z\u{a}linescu give a runtime verification framework for
metric temporal logic with freeze quantifiers over streams whose events arrive
out of order, using a three-valued semantics in which a verdict may be not yet
determined~\cite{basin2019,basin2017}. Their algorithms are sound and complete,
and POL\'IMON~\cite{polimon} continues that line with a tool that reports
verdicts promptly. The machinery this paper needs is theirs: an undetermined
verdict under partial knowledge is exactly their third value.

What a general framework leaves open is the formula. For reconciliation the property turns on something neither a temporal logic nor a stream
processor represents: which \emph{source} an absence of evidence is about. The safety
property below is not ``a match eventually occurs''. It is that silence from the
bank carries no information until the bank's own completeness horizon has
passed, and that horizon is a property of the bank rather than of the stream the
reconciler is reading.

\subsection{Reconciliation benchmarks}

FinBalance~\cite{finbalance} assembles multi-document accounting bundles with
labelled inconsistencies; FinRCA-Bench~\cite{finrca} builds
accounts-payable-to-bank cases and evaluates evidence retrieval against typed
provenance graphs. Both present a model with a complete evidence set and ask for
the right answer, and both are well suited to that question. Neither has a time
axis along which evidence is still arriving, so neither can express a verdict
that was correct but premature, which is the object of study here.

\subsection{Temporal stream models and speculation}

The CEDR line of work separates valid time, occurrence time and system time and
defines consistency levels over speculative output with
retraction~\cite{barga2007}. Retraction is property M1 below and is orthogonal
to the safety properties: a system that may retract still needs to know whether
its first verdict was justified, and in reconciliation retraction is often
unavailable because a notice has already gone out.

\subsection{Industrial practice}

Reconciliation over unbounded streams has been patented in the form of injected
window and shard sentinels that delimit an unbounded stream into reconcilable
units, with producer-defined rules carried in the closing
sentinel~\cite{patent2024}. That is an implementation technique for bounding
work, complementary to a correctness property; it states no temporal condition
on when a verdict is admissible.

Practitioner writing on exception management describes the phenomenon accurately, noting that an item inside its settlement
window is pending and not a break, and recommends holding such items until the
window expires. That instinct is
the safety property below in informal clothes. What is missing is a statement of
the property, a way to check a pipeline against it, and a measurement.

\subsection{Payment agents and effect-level evaluation}

FinalityBench~\cite{finalitybench} evaluates agents that act on payment
outcomes before those outcomes are determined, and supplies the simulator,
fault engine and delivery streams used here. Its concern and this one are the
same mistake at opposite ends of the pipeline: an agent shipping goods against
a capture that has not settled, and a reconciler raising a break against a
credit that has not posted. CausalLoss-Fin~\cite{causalloss} attributes
realised loss to infrastructure or policy after the fact using counterfactuals;
that analysis needs ground truth, where the decomposition here does not.
Operation identity, what makes two records the same money movement, follows
the provider-aware contract of~\cite{t8}.

\section{The contract}

\subsection{Setting}

Two evidence sources publish records derived from the same underlying money
movements. The \emph{ledger} reports what the business believes it is owed; the
\emph{bank} reports what actually moved. Write $\mathit{ev}(r)$ for the instant
the movement described by record $r$ occurred and $\mathit{at}(r)$ for the
instant its source published it. The two are never collapsed: every property
below relates them.

\begin{definition}[Source model]
A source model $M$ assigns to each source $s$ a lateness bound $\Lambda_s$,
meaning $\mathit{at}(r) - \mathit{ev}(r) \le \Lambda_s$ for every record $r$
from $s$, and a settlement window $W$ during which a counterpart movement may
legitimately occur. The \emph{horizon} of $s$ at time $t$ is
$H_s(t) = t - \Lambda_s$: every event of $s$ with event time at or below
$H_s(t)$ has arrived by $t$, if it was ever going to.
\end{definition}

A source model is a declaration, not a fact. An input that exceeds its bound has
broken the interface, and the contract is relative to the declaration. We return to that, because it
is what makes the check useful instead of merely strict.

An \emph{obligation} $o$ is a money movement the ledger has reported, with
identity $(\text{case}, \text{reference}, \text{kind})$ following~\cite{t8} and
event time $\mathit{ev}(o)$. At each instant a reconciler assigns every
obligation one of \textsc{pending}, \textsc{matched} or \textsc{exception}. Write
\begin{equation}
D(o) = \mathit{ev}(o) + W + \Lambda_{\mathrm{bank}}
\label{eq:decidable}
\end{equation}
for the first instant at which the bank's horizon has passed the last moment a
counterpart could have occurred.

\subsection{Four properties}

\paragraph{S1, no premature break} A reconciler may assign \textsc{exception} to
$o$ at time $t$ only if $t \ge D(o)$. Before that instant, silence from the bank
is consistent with a continuation in which the credit posts, so an exception
asserts something the evidence does not support.

\paragraph{S2, no premature match} A reconciler may assign \textsc{matched} to
$o$ at $t$ only if no admissible continuation makes that pairing wrong. A
pairing justified by identity satisfies this on arrival: nothing that arrives
later makes a record describe a different movement. A pairing justified by attributes, an amount and a nearby date, does not, because a later record
may turn out to be the real counterpart, and is admissible only once $t \ge
D(o)$.

\paragraph{L1, decisiveness} Every obligation receives a terminal verdict within
a bounded delay of $D(o)$. Without it, never deciding would satisfy both safety
properties.

\paragraph{M1, stability} A terminal verdict is not retracted.

\subsection{A worked instance}

Take the simulator's own numbers. A capture settles at tick 1000. The ledger
publishes at 1002, so an obligation exists from then. The bank posts on a
sixty-tick cycle, no sooner than fifteen ticks after the money moved, so
without knowing where in that cycle the movement fell the credit can appear anywhere in $[1015, 1074]$: as prompt as the lag, or one
tick short of a further full cycle. That range is what a reconciler may assume, and it gives
$\Lambda_{\mathrm{bank}} = 75$. With $W = 0$, \eqref{eq:decidable} gives
$D(o) = 1075$.

This movement in fact posts at 1020, twenty ticks after settling, and the
distinction matters more than it looks. The contract is stated over what a
reconciler can know, and where in the batch cycle a particular movement landed
is not among it. A reconciler that happened to be right about 1020 would have
been guessing.

A reconciler that raises a break at 1060 is wrong to do so even if the credit
never comes: at 1060 the declared bound had not expired, and a continuation
exists in which the credit posts at 1070. The verdict may turn out correct, and
it was not justified. That is the distinction the contract makes, and the reason
it can be checked without knowing the answer.

Now suppose the ledger itself was late and published at 1200. A timeout of 80 ticks measured from publication fires at 1280, which is 205
ticks after the decision became sound. The same timeout on a promptly reported movement fires at 1082.
Same bound, same bank, decisions 198 ticks apart, and the difference is entirely
attributable to a source that has nothing to do with the question being asked.

\subsection{Two consequences}

The first proposition explains a variance result rather than a mean one, and is
the reason a per-source horizon is worth the complexity.

\begin{proposition}[Ledger-lag inheritance]
\label{prop:inherit}
Let a reconciler raise an exception for $o$ at $\mathit{raised}(o) + T$, where
$\mathit{raised}(o)$ is the arrival instant of the ledger record and $T$ is a
fixed timeout. Its closure delay $\mathit{raised}(o) + T - D(o)$ equals
$\lambda(o) + T - W - \Lambda_{\mathrm{bank}}$, where
$\lambda(o) = \mathit{at}(o) - \mathit{ev}(o)$ is the realised ledger lag. The
closure delay therefore has the variance of the ledger lag distribution, for
every $T$. A reconciler that decides at $\mathit{ev}(o) + B$ for a declared
bound $B$ has closure delay $B - W - \Lambda_{\mathrm{bank}}$, a constant.
\end{proposition}

\begin{proof}
Immediate from $\mathit{raised}(o) = \mathit{ev}(o) + \lambda(o)$ and
\eqref{eq:decidable}. The second statement follows because $\mathit{ev}(o)$ is
carried in the record and does not depend on when it arrived.
\end{proof}

Raising the timeout shifts the whole distribution and does not narrow it. An
operator who wants a predictable exception queue cannot get one by tuning $T$.

\begin{proposition}[Idleness and soundness]
\label{prop:idle}
Let the bank publish in batches of period $P$, so it emits nothing for up to $P$
between batches, and let a reconciler decide from a merged watermark
$\min_s (\max \{\mathit{ev}(r) : r \text{ from } s \text{ seen}\}) - \delta$
with an idleness timeout $\tau$ that removes from the minimum any source silent
for longer than $\tau$. If $\tau < P$ there are inter-batch gaps during which
the bank is marked idle although it is behaving normally; the watermark is then
determined by the remaining sources and can exceed $H_{\mathrm{bank}}(t)$, so
S1 can be violated. If $\tau \ge P$ the bank is never marked idle merely for
being between batches.
\end{proposition}

\begin{proof}
A gap of length $g$ with $\tau < g \le P$ occurs whenever consecutive batches
are that far apart. During such a gap the bank contributes nothing to the
minimum, which is then $\max_s$ over the remaining sources less $\delta$. The
ledger publishes continuously and with smaller lateness, so its maximum event
time advances during the gap and the merged watermark advances with it, without
any bound tying it to $\Lambda_{\mathrm{bank}}$. Any obligation whose event time
falls below the advanced watermark is assigned an exception at a $t$ that need
not satisfy $t \ge D(o)$.
\end{proof}

Proposition~\ref{prop:idle} makes a specific, falsifiable prediction: violations
appear below $\tau = P$ and not above it. Section~\ref{sec:idleness} tests it,
and Section~\ref{sec:flink} tests it on Apache Flink rather than on our model
of one.

\section{Checking the contract}

\subsection{Two axes, kept apart}

The checker scores each terminal verdict twice. \emph{Admissibility} asks
whether the verdict satisfied S1 and S2 at the instant it was emitted, and is computed from the verdict, the obligation and the source model, with no
ground truth. \emph{Outcome} asks whether it was right, and needs the oracle.

Neither implies the other, and the gap is the point. A reconciler that raises a
break the moment a credit is overdue is inadmissible, and on a stream where
nothing was ever late it is also correct every time; its accuracy is a property
of the traffic. A reconciler can equally be admissible and wrong, which happens
when a source exceeds its declared bound. The checker reports those separately
and does not charge them to the reconciler, because a contract relative to a
source model that blamed the reconciler for a broken model would be
unfalsifiable.

\subsection{The contract as a query}
\label{sec:sql}

Reconciliation data is already in a database, so the contract is stated twice:
once as the predicates above, and once as SQL over three tables: arrivals, obligations and verdicts. Admissibility is a join and a comparison against
\eqref{eq:decidable}; the source-model check is an aggregate over arrivals. No
window function or recursive query is needed, which is part of the argument that
the check is cheap enough to run continuously.

M1 is deliberately not in the query. It constrains a history rather than a row,
and a per-verdict join would misstate it; retractions are counted separately.

\section{Reconcilers}

Five arms, four of them drawn from practice.

\textbf{Fixed timeout} raises a break when no counterpart arrives within $T$ of
the obligation being raised. It is what an operations team writes first, and by
Proposition~\ref{prop:inherit} its closure delay carries the ledger's jitter.

\textbf{Merged watermark} is the textbook interval join: a per-source watermark
of $\max \mathit{ev} - \delta$, and the operator takes the minimum. Sound in
ordinary configurations, and not timely, because the bank's watermark only moves
at a batch boundary.

\textbf{Merged watermark with idleness} is that job after someone fixes the
stall, with the idleness timeout of Proposition~\ref{prop:idle}.

\textbf{Amount/date matching} pairs on equal amount and proximity in time,
modelling deployments with no shared reference. Its failure is ambiguity rather
than lateness.

\textbf{Typed horizon} is the contract's reference implementation. It declares a bound per source and matches on identity, so a match is asserted
on arrival, and it waits for the \emph{bank's} horizon specifically before
raising a break.
Duplicate records are dropped by identity rather than matched again.

\section{Experimental setup}

Workloads are built from FinalityBench~\cite{finalitybench}, which supplies a
canonical event log per case, a fault engine over deliveries, and per-system
delivery streams. Its ledger trails a settlement by two ticks; its bank posts on
a 60-tick cycle no sooner than 15 ticks after the money moves, so a fault-free
fact can wait 75 ticks. That heterogeneity is the research problem and is taken
from the simulator rather than imposed.

Each regime uses five seeds of 200 interleaved cases spread over a 1440-tick
day: 1{,}733 arrivals and 885 obligations in calm traffic, 1{,}684 and 593 under
faults. Genuine breaks, movements whose bank evidence is removed outright, are
injected in 45 cases, drawn before any fault is applied, because without a
positive class a reconciler that never raises a break would score perfectly.

Injection fixes the class but does not hold it constant, and the difference
matters for how rates are reported. In calm traffic those 45 cases yield 37
graded breaks; under faults there are 210, because delivery loss destroys
further bank evidence and a credit that was dropped in transit is a real break.
Spurious breaks are therefore reported throughout as a share of the obligations the bank \emph{did} settle, 848 of them in calm
traffic and 383 under faults, instead of all obligations, so that the figure does not move merely
because the genuine-break class grew.

Every arm is calibrated on calm traffic by one rule: the smallest parameter with
no contract violations and no spurious breaks. Nobody retunes against the regime
they are scored on. Confidence intervals are percentile bootstrap over
\emph{cases}, not obligations, because obligations from one case share its
amount, archetype and fault draw; comparisons use a shared case resample.

\begin{table}[t]
\caption{Reporting lag by source, in ticks, pooled over five seeds. The gap
between the two sources is the reason a single horizon cannot serve both.}
\label{tab:lag}
\centering
\footnotesize
\begin{tabular}{llrrrr}
\toprule
& source & min & median & p99 & max \\
\midrule
\multirow{2}{*}{calm}   & ledger & 2 & 2.0 & 2.0 & 2 \\
                        & bank   & 15 & 40.0 & 73.6 & 74 \\
\midrule
\multirow{2}{*}{shifted}& ledger & 2 & 2.8 & 234.0 & 265 \\
                        & bank   & 15 & 58.0 & 283.4 & 326 \\
\bottomrule
\end{tabular}
\end{table}

Table~\ref{tab:lag} gives the lag structure. In calm traffic the ledger is
exactly two ticks behind and the bank's worst case is 74, inside the declared
bound of 75; the calibration procedure selects 80 for the arms that take a
bound directly. Under faults both tails run far past the declaration, the ledger to 265 and
the bank to 326. That is the shift the later experiments measure against.

The event store is PostgreSQL 16.4 and the stream processor is Apache Flink
1.20, both in Docker and both pinned.

\section{Results}

\subsection{What supports what}
\label{sec:evidence}

Three kinds of evidence appear below and they carry different weight, so it is
worth saying once which conclusion rests on which.

\textbf{Propositions} are proved from the definitions and hold for any system
matching them. Proposition~\ref{prop:inherit} and Proposition~\ref{prop:idle}
are of this kind, and no experiment can confirm or refute them; the
experiments show that our implementations behave as the propositions say,
which is a check on the implementations.

\textbf{Simulator results} are everything graded over reconcilers written for
this paper on workloads from the published generator. Precision, closure
delay, the frontier and the fault breakdown are all of this kind. They compare
constructions against each other on identical streams, which is a like-for-
like measurement, and they say nothing about any deployed system.

\textbf{The engine experiment} is Apache Flink 1.20 running the SQL interval
join, with its own verdicts graded by the same checker
(Section~\ref{sec:flink}). It is the only arm in which a real stream processor
makes a decision.

One boundary between the second and third matters enough to state before the
numbers rather than after. Our watermark arms agree with Flink on roughly four
obligations in five, and Section~\ref{sec:flink} gives the reason and the
figures. The timing results in the simulator are therefore properties of our
models. They are not measurements of Flink, and the Flink arm is not a
validation of them.

A second boundary concerns Proposition~\ref{prop:idle}. It is proved, and the
simulator reproduces its prediction. The engine arm does \emph{not} test it:
Flink's idleness timeout is measured in processing time and a bounded file
source is consumed in milliseconds, so the setting never engages, which is why
its four values give identical rows in Table~\ref{tab:flink}. Nothing in the
engine arm confirms the proposition and nothing in it contradicts the
proposition.

Intervals throughout are percentile bootstraps resampled over \emph{cases}
rather than obligations. Several obligations come from one sampled case and
share its amount, archetype and fault draw, so resampling obligations would
treat correlated evidence as independent and give intervals that are too
narrow. Comparisons between arms use the same case resample for both.

\subsection{Calibrated operating points}

\begin{table}[t]
\caption{Each arm calibrated on calm traffic, then run unchanged on both
regimes. ``f.break\%'' is spurious breaks as a share of the credits the bank
settled; ``exc.'' is the number of exceptions the closure statistics are
computed over. Amount/date matching never calibrates and is shown at the
largest parameter tried.}
\label{tab:main}
\centering
\footnotesize
\begin{tabular}{lrrrrrr}
\toprule
& viol.\% & f.break\% & f.match\% & exc. & mean & p99 \\
\midrule
\multicolumn{7}{l}{\emph{calm} --- 885 obligations, 848 settled, no bound violations} \\
typed horizon      & 0.00 & 0.00 & 0.00 & 37 & 2.1 & 4 \\
fixed timeout      & 0.00 & 0.00 & 0.00 & 37 & 3.8 & 6 \\
merged watermark   & 0.00 & 0.00 & 0.00 & 37 & 32.7 & 223 \\
\quad + idleness   & 0.00 & 0.00 & 0.00 & 37 & 29.9 & 201 \\
amount/date        & 94.46 & 2.12 & 22.82 & 37 & 421.5 & 426 \\
\midrule
\multicolumn{7}{l}{\emph{shifted} --- 593 obligations, 383 settled; 240 bank and} \\
\multicolumn{7}{l}{413 ledger bound violations} \\
typed horizon      & 0.17 & 22.98 & 0.00 & 298 & 12.9 & 146 \\
fixed timeout      & 0.00 & 20.89 & 0.00 & 290 & 28.9 & 226 \\
merged watermark   & 0.51 & 18.28 & 0.00 & 280 & 53.0 & 213 \\
\quad + idleness   & 1.01 & 18.80 & 0.00 & 282 & 45.5 & 160 \\
amount/date        & 47.05 & 8.36 & 32.72 & 141 & 458.5 & 637 \\
\bottomrule
\end{tabular}
\end{table}

Table~\ref{tab:main} gives the operating points. On calm traffic the four
identity-based arms are all exactly sound, and they differ by more than an order
of magnitude in how quickly a genuine break surfaces: 2.1 ticks for the typed
horizon against 32.7 for the merged watermark, and a tail of 4 ticks against
223. Only 37 exceptions arise in calm traffic, every one of them a genuine break,
so the calm percentile is close to a maximum and should be read as
one; the shifted regime, where the statistics rest on roughly 290 exceptions
per arm, is the one that supports a percentile claim. The watermark arms are not slow because they are cautious. They are slow
because the bank's watermark only advances at a batch boundary, so a break
discovered at the start of a cycle waits for the end of it.

Amount/date matching never calibrates. There is no parameter at which it is both
sound and accurate, because its failure is not a matter of waiting: with eight
distinct retail amounts in circulation, two movements of the same size in flight
together are indistinguishable, and it pairs the wrong records on 22.82\% of
obligations in calm traffic. This is the sense in which matching rules are downstream
of the temporal contract. No amount of tuning the decision deadline repairs a
pairing that was never justified.

\subsection{Soundness against timeliness}
\label{sec:frontier}

\begin{figure}[t]
\centering
\includegraphics[width=\columnwidth]{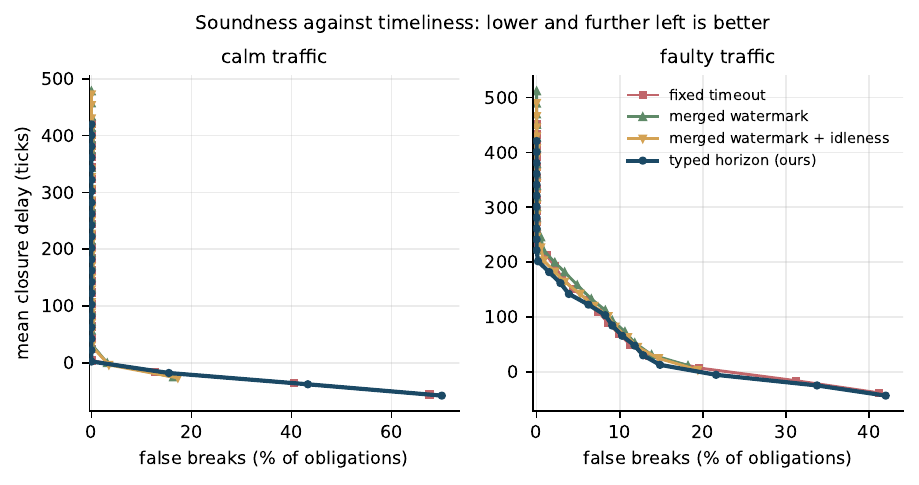}
\caption{Sweeping each family's parameter. Down and to the left is better.}
\label{fig:frontier}
\end{figure}

\begin{table}[t]
\caption{Cheapest point on each curve with no spurious breaks, shifted traffic.
Brackets are 95\% bootstrap intervals over cases.}
\label{tab:frontier}
\centering
\footnotesize
\begin{tabular}{lrrr}
\toprule
& mean closure & p99 & s.d. \\
\midrule
typed horizon     & 221.6 [220.8, 222.1] & 224 & 5.0 \\
\quad + idleness  & 245.8 [240.6, 251.4] & 392 & 40.4 \\
fixed timeout     & 253.8 [245.7, 262.5] & 446 & 60.8 \\
merged watermark  & 264.9 [255.8, 274.2] & 554 & 69.8 \\
\bottomrule
\end{tabular}
\end{table}

Comparing arms at their calibrated parameters is what an operator would deploy
but is not a controlled comparison. Figure~\ref{fig:frontier} sweeps each
family; Table~\ref{tab:frontier} reads off the cheapest point at which no
spurious break is raised, which is the operating point reconciliation teams
actually ask for.

\begin{figure}[t]
\centering
\includegraphics[width=0.86\columnwidth]{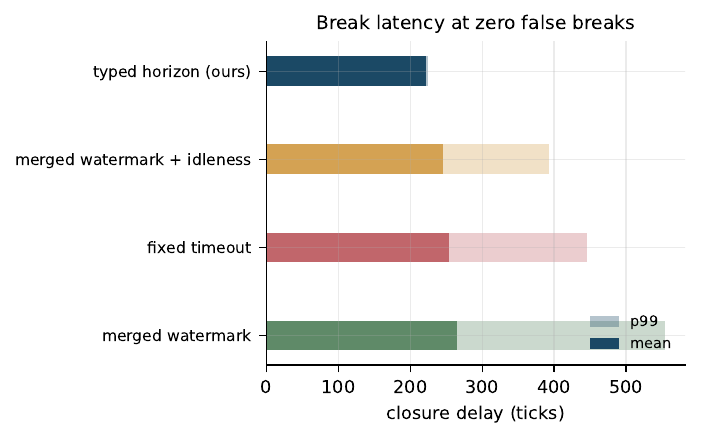}
\caption{Mean and 99th-percentile break latency at the zero-false-break
operating point, shifted traffic. The bars for the typed horizon nearly
coincide; for the others they do not.}
\label{fig:tail}
\end{figure}

The typed horizon reaches that point 12.7\% sooner than a fixed timeout, with
non-overlapping intervals, and its tail is halved.
Figure~\ref{fig:tail} shows why that matters operationally: for the typed
horizon the mean and the 99th percentile are almost the same number, so the
exception queue drains predictably, while for the others the worst decile
arrives at roughly twice the average delay. The standard deviation is the
clearest number: 5.0 ticks against 60.8 and 69.8. That is
Proposition~\ref{prop:inherit} showing up in the data. The other arms inherit
the jitter of the ledger's reporting into a decision that is about the bank,
and raising their parameter moves the distribution without narrowing it.

\subsection{What happens when the distribution moves}

\begin{figure}[t]
\centering
\includegraphics[width=\columnwidth]{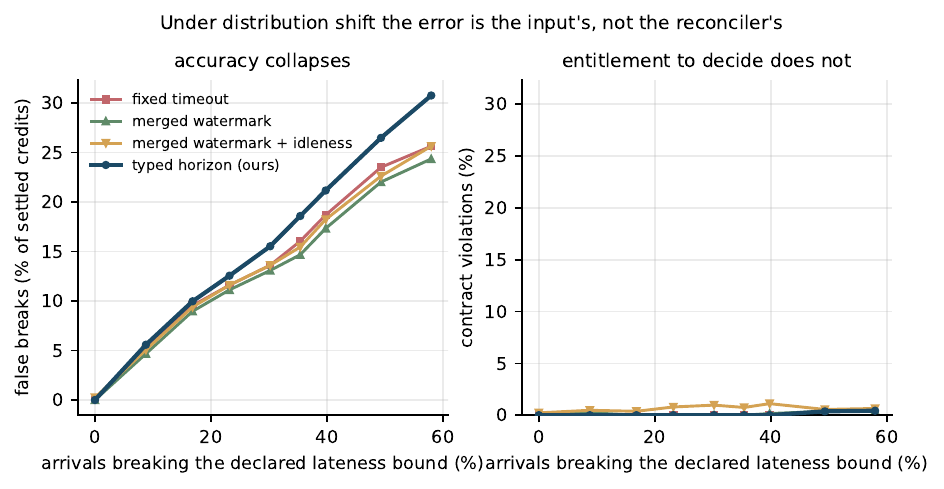}
\caption{As fault intensity rises, accuracy collapses and entitlement does not.}
\label{fig:shift}
\end{figure}

Figure~\ref{fig:shift} sweeps fault intensity from zero to 0.5 on seeds disjoint
from the calibration set. The share of arrivals breaking their declared bound
rises from 0\% to 57.9\%, and spurious breaks rise with it, reaching 25.6\% of
settled credits for the fixed timeout and 30.8\% for the typed horizon.

The contract violation rate barely moves. Across the whole sweep it stays at or
below 1.11\% for every arm, and at or below 0.43\% for every arm that does not use an idleness timeout. The exception is the one
construction we already know breaks S1. These reconcilers were, with that caveat, entitled to the
verdicts they issued, and they were wrong anyway, because the inputs stopped honouring the bounds
against which they were calibrated. The checker
reaches that conclusion from the arrivals and the declared model, without
consulting the oracle, which is what makes it usable on the morning the break
rate triples.

This is the practical content of separating the two axes. An operator watching
only the break rate sees a reconciler that has become unreliable and retunes it.
The retuning cannot help: no timeout repairs a bank that has started posting
late, and lengthening it trades a real detection delay for an apparent
improvement.

\subsection{The idleness trap}
\label{sec:idleness}

\begin{figure}[t]
\centering
\includegraphics[width=0.86\columnwidth]{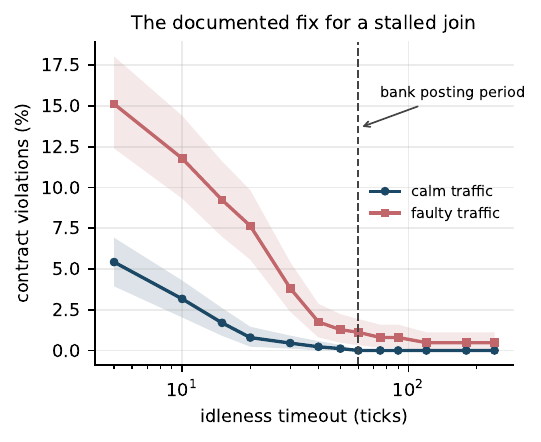}
\caption{Contract violations against the idleness timeout, with the bank's
posting period marked.}
\label{fig:idle}
\end{figure}

With idleness disabled the merged watermark commits no contract violations on
calm traffic at all. Turning it on, at a timeout of 5 ticks, produces violations
on 5.42\% of obligations (95\% CI [3.94, 6.92]) and spurious breaks on 3.05\% ---
on traffic in which no source is ever late and every credit arrives.

Figure~\ref{fig:idle} shows the cliff. Violations fall monotonically as the
timeout rises --- 0.45\% at 30 ticks, 0.11\% at 50 --- and reach exactly zero at
60, the bank's posting period, remaining there above it. That is the boundary
Proposition~\ref{prop:idle} predicts, and it was predicted before it was
measured.

The trap is that the correct setting defeats the purpose. Idleness is reached
for because the join stalls between batches; the smallest timeout that keeps the
join sound is the batch period, which is precisely the gap the operator was
trying to stop waiting through.

\subsection{Which faults matter}

Enabling one fault family at a time is unusually decisive here. Of the six
delivery-fault families in the simulator, only two touch reconciliation
correctness at all. Delay produces 462 bound violations and drives spurious
breaks to 22.9\% for the typed horizon. Duplication produces 430 bound
violations and, for identity-based arms, no errors whatsoever, because a repeated delivery of the same movement is recognised and dropped.
Amount/date matching mispairs 36.4\% of its verdicts under the same faults.
Reordering, partial commits and stale reads leave the ledger-to-bank path
untouched, with results identical to the fault-free arm to three digits.

\begin{table}[t]
\caption{One fault family at a time, at intensity 0.30. ``bound viol.'' counts
arrivals exceeding their declared lateness. Only two families reach the
ledger-to-bank path at all.}
\label{tab:families}
\centering
\footnotesize
\begin{tabular}{lrrrr}
\toprule
& & \multicolumn{2}{c}{typed horizon} & amount/date \\
\cmidrule(lr){3-4}\cmidrule(lr){5-5}
family & bound viol. & f.break\% & f.match\% & f.match\% \\
\midrule
none            &   0 &  0.0 & 0.0 & 24.0 \\
delay           & 462 & 22.9 & 0.0 & 35.9 \\
duplicate       & 430 &  0.0 & 0.0 & 36.4 \\
loss            &   0 &  0.0 & 0.0 & 28.9 \\
reorder         &   0 &  0.0 & 0.0 & 24.0 \\
partial commit  &   0 &  0.0 & 0.0 & 24.0 \\
stale read      &   0 &  0.0 & 0.0 & 24.0 \\
\bottomrule
\end{tabular}
\end{table}

Table~\ref{tab:families} gives the breakdown.
Reconciliation correctness has a narrow failure surface: lateness and
repetition. Lateness is a safety problem for everyone; repetition is a safety
problem only for reconcilers that cannot tell two records apart.

\subsection{Cost, and where our construction loses}

The monetary figures rest on three constants argued rather than measured, so
each is swept over an order of magnitude, giving 125 combinations per regime. On
calm traffic the typed horizon is cheapest at all 125. Under shifted traffic it
is cheapest at 30, and the merged watermark is cheapest at 75.

\begin{table}[t]
\caption{Cheapest arm over a $5\times5\times5$ sweep of the three cost
constants, 125 combinations per regime. The reversal between regimes is the
result, not the dollar figures.}
\label{tab:cost}
\centering
\footnotesize
\begin{tabular}{lrr}
\toprule
cheapest at & calm & shifted \\
\midrule
typed horizon    & 125 & 30 \\
merged watermark &   0 & 75 \\
fixed timeout    &   0 & 20 \\
\quad + idleness &   0 &  0 \\
amount/date      &   0 &  0 \\
\bottomrule
\end{tabular}
\end{table}

That reversal, in Table~\ref{tab:cost}, is real and worth stating plainly. At its calibrated parameter the
typed horizon decides at the earliest instant its declared bound permits, which is exactly what makes it fast, and exactly what exposes it when the
bound turns out to be wrong. The merged watermark's extra conservatism, which costs it
an order of magnitude in latency on calm traffic, is what pays under shift. At
matched soundness (Table~\ref{tab:frontier}) the ordering reverses again.

The contract does not make a reconciler cheap. It makes the trade explicit, and
it makes the residual attributable.

\subsection{Checking the contract in SQL}

Across five arms, three seeds and both regimes, the SQL formulation of
Section~\ref{sec:sql} and the reference implementation agree on all 4{,}395
verdicts compared, on both admissibility and outcome, with no disagreements.
The check is a join and a comparison, so a reconciliation warehouse can run it
against yesterday's verdicts without new infrastructure.

\subsection{Against a real stream processor}
\label{sec:flink}

Everything so far grades reconcilers written for this paper. The same
workloads were also replayed through Apache Flink 1.20, as an ordinary LEFT interval join between the two streams, the shape the
pattern is usually built as, with the engine's own verdicts graded by the same
checker. The decision
instant is recovered from \texttt{CURRENT\_WATERMARK}, which the job selects
alongside each row, mapped back to the wall-clock instant that watermark could
first have been reached.

\begin{table}[t]
\caption{Apache Flink 1.20, SQL interval join, shifted traffic. Each row
averages eight runs, four idleness settings by two seeds; the idleness timeout
changes nothing here, for the reason given in the text, so the spread is the
engine disagreeing with itself.}
\label{tab:flink}
\centering
\footnotesize
\begin{tabular}{rrlrrr}
\toprule
window & viol.\% & (range) & f.break\% & mean & p99 \\
\midrule
80  & 14.91 & 14.78--15.04 & 18.28 & 51.5 & 239 \\
160 & 18.18 & 14.16--20.87 & 16.00 & 30.1 & 176 \\
320 & 20.09 & 16.52--23.01 & 14.80 & 18.3 & 161 \\
\bottomrule
\end{tabular}
\end{table}

The engine violates S1 on 14.91\% to 20.09\% of obligations, averaged over the
eight runs of each window in Table~\ref{tab:flink}; individual runs span
14.16\% to 23.01\%. It does so in the configuration a competent engineer would
write.
Widening the join interval buys timeliness and costs soundness, which is the
frontier of Section~\ref{sec:frontier} seen from inside a real system. The
engine also emits one output row per matching pair, so a duplicated bank record produces a duplicated match, 15 to 44 extra rows
per run, which a downstream sink has to deduplicate on its own.

The window effect is solid; the individual rates are less so. Window 80 is
identical to the second decimal place on both seeds across all four idleness
settings. Windows 160 and 320 are not: at each, one seed repeats exactly while
the other moves across settings that cannot affect it, by 1.77 and 0.87 points
respectively. Nothing in the configuration distinguishes those cells, so the
spread is the engine disagreeing with itself, and which cell it lands in is not
a property of the window or the seed. The ordering across windows survives
every run we have made; a single per-configuration rate does not. What follows
is less comfortable.

\textbf{Our watermark arms are not a faithful model of the engine.} The
simulated merged watermark agrees with Flink's partition of the obligations on
77.0\% and 78.8\% of them, across the two seeds. Roughly a fifth of the time the
model and the engine reach different conclusions, and the explanation is
structural rather than statistical: the simulated arm models the \emph{decision rule}, raise a break once a watermark
has passed, while the engine's interval join also constrains which bank records are eligible to match at all, through
the event-time interval in the join predicate. Two different constructions,
agreeing most of the time. The timing results in Table~\ref{tab:frontier} are
therefore properties of the models, and we state them as such rather than as
measurements of Flink.

\textbf{The instability rate is itself a random variable.} An earlier full run
of this experiment, made while the machine was also running the other seven,
showed the engine disagreeing with itself: repeating an identical job four
times, 26.1\% of obligations on one seed and 40.7\% on the other received
different verdicts. Re-run on an otherwise idle machine, the same configuration
gave 0.9\% and 0.0\%. We withdrew the observation, and the obvious reading of
the withdrawal --- that load explained it --- is one we can now test rather than
assert. E10 varies load deliberately, running the determinism check idle and
under twelve busy-loop workers on twelve cores, six times per condition.

Load does not explain it. Across six measurements per condition the idle arm
spans 0.0\% to 11.30\% and the loaded arm 0.0\% to 7.83\%; the condition means
differ by 0.43 points, an order of magnitude inside the spread of either
condition, and in the direction opposite to the hypothesis. The distributions
overlap almost completely, and the largest idle reading exceeds every one of
the six loaded ones. The useful result is the spread, not the comparison:
repeating an identical job on an unchanged machine gives a verdict-instability
rate anywhere between zero and a tenth of all obligations, a range spanning
more than an order of magnitude within one condition, so no single run --- the
loaded one we withdrew, the idle one we replaced it with, or any of the twelve
reported here --- carries information about the engine on its own. The E7 run
shipped with this paper happens to be a quiet one, with seed~1 disagreeing with
itself on 3.48\% of obligations, and that is not evidence of stability either.
It is why Table~\ref{tab:flink} reports a range per window rather than a figure
per configuration.

\textbf{What that does and does not undermine.} The volatility is confined to
the repeat-to-repeat determinism check. The violation rates underlying
Table~\ref{tab:frontier} are averages over 200 obligations rather than counts of
disagreement between four repeats, and they hold across whole-experiment runs
made days apart. The model-fidelity figures sit between the two: reported per
seed from one execution, and E10's spread is a reason to read them as
approximate. Where this paper quotes a determinism or fidelity number we mean
\emph{this run gave this}, not \emph{the engine does this} --- so a reviewer who
runs our determinism check once and gets a number unlike ours has reproduced us
correctly.

One thing E7 could not test at all. Flink's idleness timeout is measured in
processing time, and a bounded file source is consumed in milliseconds, so no source is ever idle long enough for the setting to engage. That is why
three of its four values give identical numbers above. Proposition~\ref{prop:idle} is
supported by its proof and by the model in Section~\ref{sec:idleness}, and is
\emph{not} confirmed on the engine here. Testing it would need a source paced in
processing time, which means a broker and a producer rather than a file, and
that is left as future work.

Running the same job with \texttt{execution.runtime-mode = BATCH} produces no
violations at all, because a batch job sees its whole input and so has no
partial-evidence problem to get wrong. Offline reconciliation is sound, and the
price is that it has no answer until the batch runs. Everything in this paper is
about the systems built to avoid paying that price.

\section{Discussion}

\paragraph{A horizon is not a watermark} A watermark summarises what a stream
processor has seen. A horizon states what a source promised. They coincide when
one source is in play and diverge as soon as two sources have different
latencies, because a merged watermark answers ``how far has the input
progressed'' when the question is ``could the bank still be about to tell me''.
Most of the failures in this paper are that substitution.

\paragraph{Declaring a bound is the useful part} The typed horizon is not a
sophisticated reconciler. It waits for one number and matches on identity. Its
advantage is that the number is written down, so when it is wrong the checker
says so and names the source. Every other arm encodes the same assumption
implicitly in a timeout, where it cannot be checked and cannot be blamed.

\paragraph{The check is cheap where the data already is} Nothing in S1 or S2
requires the reconciler's cooperation. Both are predicates over a verdict log,
an obligation table and a declared model, which is why
Section~\ref{sec:sql} can state them as a join. An organisation that cannot change its reconciliation engine, the common case
since that engine is often a vendor product, can still run the check against yesterday's output
and learn whether its break queue is the engine's doing or its counterparties'.

\paragraph{Matching is downstream} A reconciliation product is usually sold on
its matching rules. The results here put that second: attribute matching cannot
satisfy S2 at any deadline, and identity matching satisfies it immediately. The
hard part is not deciding which records correspond. It is deciding when you are
allowed to conclude that none does.

\section{Limitations}

Two sources, not many. Reconciling three or more generalises the contract, with the horizon becoming
a vector, but is not tested, and multi-way reconciliation may
raise ambiguities that two-way does not.

Bounds are declared, not learned. Estimating $\Lambda_s$ online, and stating a
confidence rather than a guarantee, is the obvious next step and would turn S1
into a probabilistic property. We chose the bounded form because it is checkable
exactly.

Obligations are ledger-side. A bank credit with no ledger expectation is
counted and reported but not graded, and that direction of break is a real
concern we do not measure.

The workload is synthetic throughout, built from one simulator. Its lateness structure, a fast ledger and a batching bank, is realistic in
shape, and the
contract does not depend on the particular numbers, but the rates reported here
are properties of this generator and should not be read as field rates.

The cost constants are argued, not measured. Where the ordering of arms depends
on them, as it does under shift, we report the dependence instead of a figure.

\section{Conclusion}

Reconciliation on a stream is a decision under partial evidence, and the
question of when the evidence is sufficient has been answered in deployed
systems by a constant nobody writes down. Writing it down as a per-source lateness bound, with a safety property
relating it to the instant a verdict is issued, costs little and buys a
diagnosis: contract violation, model
violation, or genuine loss, told apart without knowing the answer. The parts of
this that surprised us were that the documented remedy for a stalled streaming
join is unsound for reconciliation below the slow source's posting period, and
that under distribution shift every reconciler we tested was entitled to every
verdict it got wrong.

\section*{Artifact availability}

Code, workload generator, Docker configuration and every results file are at
\url{https://github.com/abhisheksharma2411/reconstream}, MIT licensed, tagged
\texttt{v1.0.2}, which is the revision this paper was built from. The release is archived at DOI \texttt{10.5281/zenodo.23272623}.
\texttt{make reproduce} runs the whole
pipeline from a clean checkout: it fetches FinalityBench, starts PostgreSQL and
Flink, runs every experiment, regenerates each figure from the raw results and
rebuilds this paper. No number in the text is entered by hand.

\balance

\begin{thebibliography}{99}

\bibitem{basin2019}
D.~Basin, F.~Klaedtke, and E.~Z\u{a}linescu, ``Runtime verification over
out-of-order streams,'' \emph{ACM Trans. Comput. Logic}, vol.~20, no.~3, 2019.
doi:10.1145/3355609. Preprint arXiv:1909.11593.

\bibitem{basin2017}
D.~Basin, F.~Klaedtke, and E.~Z\u{a}linescu, ``Runtime verification of temporal
properties over out-of-order data streams,'' in \emph{Computer Aided
Verification (CAV)}, LNCS 10426, Springer, 2017, pp.~356--376.

\bibitem{polimon}
F.~Klaedtke, ``POL\'IMON: Checking temporal properties over out-of-order
streams at runtime,'' arXiv:2404.15723, Apr.~2024.

\bibitem{finbalance}
S.~Tumpati, D.~Agarwal, A.~Kedia, A.~Neekhra, M.~Mandal, K.~Garg, Y.~Sinha,
S.~Gupta, and D.~Kumar, ``FinBalance: A multi-document accounting
reconciliation benchmark,'' arXiv:2606.15949, Jun.~2026.

\bibitem{finrca}
P.~Ghawate, ``FinRCA-Bench: Benchmarking evidence retrieval and reasoning for
financial AI systems,'' arXiv:2608.18534, Aug.~2026.

\bibitem{barga2007}
R.~S. Barga, J.~Goldstein, M.~Ali, and M.~Hong, ``Consistent streaming through
time: A vision for event stream processing,'' in \emph{Conf. Innovative Data
Systems Research (CIDR)}, 2007, pp.~363--374. Preprint arXiv:cs/0612115.

\bibitem{patent2024}
D.~Stewart and A.~Dwivedi, ``Reconciliation systems and methods for unbounded
streams,'' U.S. Patent 12,045,262 B2, Jul.~23, 2024.

\bibitem{finalitybench}
A.~Sharma, ``FinalityBench: An effect-level benchmark for agent decisions under
delayed and conflicting financial finality,'' arXiv:2609.04706, 2026,
\textsc{doi}: 10.5281/zenodo.22262591.

\bibitem{causalloss}
A.~Sharma, ``CausalLoss-Fin: Attributing financial-agent loss to decisions and
infrastructure faults,'' arXiv:2609.25960, 2026,
\textsc{doi}: 10.5281/zenodo.22893020.

\bibitem{t8}
A.~Sharma, ``A provider-aware formal contract and conformance suite for payment
idempotency,'' under review, 2026.

\end{thebibliography}
\end{document}